\documentclass[submission,copyright,creativecommons]{eptcs}
\providecommand{\event}{TARK 2015} % Name of the event you are submitting to
\usepackage{breakurl}             % Not needed if you use pdflatex only.

\usepackage[T1]{fontenc}

\usepackage{latexsym}

\usepackage{mathtools}

\usepackage{amsopn}

\usepackage{graphics}
\usepackage{tikz}
\usepackage[all]{xy}
\usepackage{url}
\usepackage{enumerate}

\usepackage{amsmath}
\usepackage{amssymb}
\usepackage{xspace}
\usetikzlibrary{arrows,backgrounds,positioning}
\usetikzlibrary{snakes}

\DeclareMathOperator{\isPos}{pos}

\newcommand{\game}{\mathcal{G}_\phi}

\newcommand{\m}{\mathcal}

\newcommand{\SW}{\textit{SW}}

\newcommand{\calg}{\mathcal{G}}
\newcommand{\cpath}{\mathit{CPath}}
\newcommand{\lasts}{\mathit{last}}

\newtheorem{theorem}{Theorem}

\newtheorem{lemma}[theorem]{Lemma}
\newenvironment{proof}{\noindent{\bf Proof.}}{\hfill\qed\smallskip}
\newtheorem {example}[theorem]{Example}
\newtheorem{corollary}[theorem]{Corollary}

\newcommand {\qed} {\hfill \ensuremath{\Box} \\}

\newcommand{\payoff}{p}
\newcommand{\strprofile}{s}

\newcommand{\betredge}[1]{\betstep{#1}}
\newcommand{\betstep}[1]{\mbox{$\stackrel{#1}{\rightarrow}$}}

\newcounter{symbol}
\newcommand{\indexsyma}[1]%
{\stepcounter{symbol}\index{zzz1 \thesymbol @\protect#1}}
\newcommand{\indexsymb}[1]%
{\stepcounter{symbol}\index{zzz2 \thesymbol @\protect#1}}
\newcommand{\indexsymc}[1]%
{\stepcounter{symbol}\index{zzz3 \thesymbol @\protect#1}}
\newcommand{\indexsymd}[1]%
{\stepcounter{symbol}\index{zzz4 \thesymbol @\protect#1}}
\newcommand{\indexsyme}[1]%
{\stepcounter{symbol}\index{zzz5 \thesymbol @\protect#1}}

\newcommand{\bfe}[1]{\begin{bfseries}\emph{#1}\end{bfseries}\index{#1}}

\newcommand{\myra}{\mbox{$\:\rightarrow\:$}}

\newcommand{\sse}{\mbox{$\:\subseteq\:$}}

\newcommand{\LL}{\mbox{$\ldots$}}

\newcommand{\C}[1]{\mbox{$\{{#1}\}$}}           % curly braces

\newcommand{\NI}{\noindent}
\newcommand{\HB}{\hfill{$\Box$}}
\newcommand{\VV}{\vspace{5 mm}}

\newcommand{\szkew}[1]{\relax \setbox0=\hbox{\kern -24pt $\displaystyle#1$\kern 0pt }%
\box0}
{\catcode`\@=11 \global\let\ifjusthvtest@=\iffalse}
\newcounter{oldmycaption}

\title{Coordination Games on Directed Graphs}
\author{Krzysztof R. Apt
\institute{Centrum Wiskunde \& Informatica}
\institute{Amsterdam, The Netherlands}
\email{k.r.apt@cwi.nl}
\and
Sunil Simon
\institute{Department of CSE}
\institute{IIT Kanpur, Kanpur, India}
\email{simon@cse.iitk.ac.in}
\and
Dominik Wojtczak
\institute{University of Liverpool}
\institute{Liverpool, U.K.}
\email{d.wojtczak@liv.ac.uk}
}
\def\titlerunning{Coordination Games on Directed Graphs}
\def\authorrunning{K. R. Apt, S. Simon \& D. Wojtczak}

\begin{document}
\maketitle

\begin{abstract}
  We study natural strategic games on directed graphs, which capture
  the idea of coordination in the absence of globally common
  strategies. We show that these games do not need to have a pure Nash
  equilibrium and that the problem of determining their existence is
  NP-complete.  The same holds for strong equilibria.  We also exhibit
  some classes of games for which strong equilibria exist and prove
  that a strong equilibrium can then be found in linear time.
\end{abstract}
\date{}

\section{Introduction}

In this paper we study a simple and natural class of strategic games.
Assume a finite directed graph. Suppose that each node selects a
colour from a private set of colours available for it. The payoff to a
node is the number of (in)neighbours who chose the same colour.

These games are typical examples of coordination games.  Recall that
the idea behind \emph{coordination} in strategic games is that players
are rewarded for choosing common strategies.  The games we study here are
specific coordination games in the absence of globally common strategies.

Recently, we studied in \cite{ARSS14}, and more fully in
\cite{ARSS15}, a very similar class of games in which the graphs were
assumed to be undirected. However, the transition from undirected to
directed graphs drastically changes the status of the games.  For
instance, for the case of directed graphs Nash equilibria do not need
to exist, while they always exist when the graph is undirected.
Consequently, in \cite{ARSS14} and \cite{ARSS15} we focused on the
problem of existence of strong equilibria. We also argued there that
such games are of relevance for the cluster analysis, the task of
which is to partition in a meaningful way the nodes of a graph.  The
same applies here. Indeed, once the strategies are possible cluster
names, a Nash equilibrium naturally corresponds to a `satisfactory'
clustering of the underlying graph.

The above two classes of games are also similar in that both
are special cases of a number of well-studied types of games.  One
of them are \bfe{polymatrix games} introduced in \cite{Jan68}. In
these games the payoff for each player is the sum of the payoffs from
the individual two player games he plays with each other player
separately.  Another are \bfe{graphical games} introduced in
\cite{KLS01}. In these games the payoff of each player depends only on
the strategies of its neighbours in a given in advance graph structure
over the set of players.

In addition both classes of games satisfy the \bfe{positive population
  monotonicity} (PPM) property introduced in \cite{KBW97} that states
that the payoff of each player weakly increases if another player
switches to his strategy.  Coordination games on graphs are examples
of games on networks, a vast research area recently surveyed in
\cite{JZ14}.  Other related references can be found in \cite{ARSS15}.

\subsection{Plan of the paper and the results}

In the next section we introduce preliminary definitions, following
\cite{ARSS15}.  We define the coordination games on directed graphs in
Section \ref{sec:colouring}. In Section \ref{sec:directed-strong} we
exhibit a number of cases when a strong equilibrium exists.  Next, in
Section \ref{sec:directed-complexity} we study complexity of the problem
of existence of Nash and strong equilibria and the problem of
determining the complexity of finding a strong equilibrium in a
natural case when it is known to exist.  Finally, in Section
\ref{sec:conc} we discuss future directions.

The main results are as follows.  If the underlying graph is a DAG, is
complete or is such that every strongly connected component (SCC) is a
simple cycle, then strong equilibria always exist and they can always
be reached from any initial joint strategy by a sequence of
coalitional improvement steps.  The same is the case when only two
colours are used.

In general Nash equilibria do not need to exist and the problem
of determining their existence is NP-complete.  The same is the case
for strong equilibria.  We also show that when every SCC is a simple
cycle, then strong equilibrium can always be found in linear time.

\section{Preliminaries}
\label{sec:prelim}

%% A \bfe{strategic game} $\mathcal{G}$ for $n > 1$ players, written as 
%% $(S_1, \ldots, S_n,$ $p_1, \ldots, p_n)$, consists of a non-empty set
%% $S_i$ of \bfe{strategies} and a \bfe{payoff function} $p_i : S_1
%% \times \cdots \times S_n \myra \mathbb{R}$, for each player $i$.

A \bfe{strategic game} $\mathcal{G}=(S_1, \ldots, S_n,$ $p_1, \ldots,
p_n)$ with $n > 1$ players, consists of a non-empty set $S_i$ of
\bfe{strategies} and a \bfe{payoff function} $p_i : S_1 \times \cdots
\times S_n \myra \mathbb{R}$, for each player $i$.

We denote $S_1 \times \cdots \times S_n$ by $S$, call each element $s
\in S$ a \bfe{joint strategy} and abbreviate the sequence $(s_{j})_{j
  \neq i}$ to $s_{-i}$. Occasionally we write $(s_i, s_{-i})$ instead
of $s$.
We call a strategy $s_i$ of player $i$ a \bfe{best response} to a
joint strategy $s_{-i}$ of his opponents if for all $ s'_i \in S_i$,
$p_i(s_i, s_{-i}) \geq p_i(s'_i, s_{-i})$. 
% A joint strategy $s$ is
% called a \bfe{Nash equilibrium} if each $s_i$ is a best response to
% $s_{-i}$.

Fix a strategic game $\mathcal{G}$.  We say that $\mathcal{G}$
satisfies the \bfe{positive population monotonicity (PPM)} if for all
joint strategies $s$ and players $i, j$, $p_i(s) \leq p_i(s_i,
s_{-j})$.  (Note that $(s_i, s_{-j})$ refers to the joint strategy in
which player $j$ chooses $s_i$.) So if more players (here just player
$j$) choose player $i$'s strategy and the remaining players do not
change their strategies, then $i$'s payoff weakly increases.

We call a non-empty subset $K := \{k_1, \ldots, k_m\}$ of the
set of players $N:= \{1, \ldots, n\}$ a
\bfe{coalition}. Given a joint strategy $s$ we abbreviate the
sequence $(s_{k_1}, \ldots, s_{k_m})$ of strategies to $s_K$ and
$S_{k_1} \times \cdots \times S_{k_m}$ to $S_{K}$. We also write
$(s_K, s_{-K})$ instead of $s$. If there is a strategy $x$ such that
$s_i = x$ for all players $i \in K$, we also write $(x_K, s_{-K})$ for
$s.$

Given two joint strategies $s'$ and $s$ and a coalition $K$, we say
that $s'$ is a \bfe{deviation of the players in $K$} from $s$ if $K =
\{i \in N \mid s_i \neq s_i'\}$.  We denote this by $s \betredge{K}
s'$. If in addition $p_i(s') > p_i(s)$ holds for all $i \in K$, we say
that the deviation $s'$ from $s$ is \bfe{profitable}. Further, we say
that the players in $K$ \bfe{can profitably deviate from $s$} if
there exists a profitable deviation of these players from $s$.

Next, we call a joint strategy $s$ a \bfe{k-equilibrium}, where $k
\in \{1, \dots, n\}$, if no coalition of at most $k$ players can
profitably deviate from $s$.  Using this definition, a \bfe{Nash
  equilibrium} is a 1-equilibrium and a \bfe{strong equilibrium}, see
\cite{Aumann59}, is an $n$-equilibrium.

Given a joint strategy $s$, the \bfe{social welfare} of $s$ is defined as, 
%% we call the sum 
\[\mathit{\SW}(s)=\sum_{i \in N} p_i(s).\]
%%the \bfe{social welfare} of $s$. 

%% \marginpar{TO DELETE?}

%% When the social
%% welfare of $s$ is maximal we call $s$ a \bfe{social optimum}.  Given
%% a finite game that has a $k$-equilibrium its \bfe{$k$-price of
%%   anarchy (resp. stability)} is the ratio
%% $\mathit{\SW}(s)/\mathit{\SW}(s')$, where $s$ is a social optimum and
%% $s'$ is a $k$-equilibrium with the lowest (resp. highest) social
%% welfare. In the case of division by zero, we interpret the outcomes as
%% $\infty$. The \bfe{(strong) price of anarchy} refers to the $k$-price
%% of anarchy with $k = 1$ ($k = n$). The \bfe{(strong) price of
%%   stability} is defined analogously.
%% \marginpar{END: TO DELETE?}

A \bfe{coalitional improvement path} (\textit{c-improvement path}), is
a maximal sequence $\rho=(s^1, s^2, \dots)$ of joint strategies such
that for every $k > 1$ there is a coalition $K$ such that $s^k$ is a
profitable deviation of the players in $K$ from $s^{k-1}$. If $\rho$
is finite then by $\lasts(\rho)$ we denote the last element of the
sequence.  Clearly, if a c-improvement path is finite, its last
element is a strong equilibrium. We say that $\mathcal{G}$ has the
\bfe{finite c-improvement property} (\bfe{c-FIP}) if every
c-improvement path is finite. Further, we say that the function $P: S
\rightarrow A$, where $A$ is a set, is a \bfe{generalized ordinal
  c-potential}, also called \bfe{generalized strong potential}, see
\cite{Harks13,Holzman97}, for $\m G$ if for some strict partial
ordering $(P(S), \succ)$ the fact that $s'$ is a profitable deviation
of the players in some coalition from $s$ implies that $P(s') \succ
P(s)$.

If a finite game admits a generalized ordinal c-potential then it has
the c-FIP.  The converse also holds, see, e.g., \cite{ARSS15}.  We say
that $\mathcal{G}$ is \bfe{c-weakly acyclic} if for every
joint strategy there exists a finite c-improvement path that starts at
it. Note that games that have the c-FIP or are c-weakly acyclic
game have a strong equilibrium.

We call a c-improvement path an \bfe{improvement path} if each
deviating coalition consists of one player. The notions of a game
having the \bfe{FIP} or being \bfe{weakly acyclic}, see
\cite{You93,Mil96}, are then defined by referring to improvement paths
instead of c-improvement paths.

\section{Coordination games on directed graphs}
\label{sec:colouring}

We now introduce the class of games we are interested in.  Fix a
finite set of colours $M$ and a weighted directed graph $(G,w)$ without self
loops in which each edge $e$ has a non-negative weight $w_e$
associated with.  We say that a node $j$ is a \bfe{neighbour} of the
node $i$ if there is an edge $j \to i$ in $G$.  Let $N_i$ denote the
set of all neighbours of node $i$ in the graph $G$.  By a \bfe{colour
  assignment} we mean a function that assigns to each node of $G$ a
finite non-empty set of colours.  For technical reasons we also
introduce the concept of a \bfe{bonus}, which is a function $\beta$
that to each node $i$ and a colour $c$ assigns a natural number
$\beta(i,c)$. (We allow zero as a natural number.)

Given a weighted graph $(G,w)$, a colour assignment $A$ and a bonus function $\beta$
we define a strategic game $\mathcal{G}(G,w,A,\beta)$ as
follows:

\begin{itemize}
\item the players are the nodes,
  
\item the set of strategies of player (node) $i$ is the set of colours
  $A(i)$; we refer to the strategies as \bfe{colours} and to joint
  strategies as \bfe{colourings},

\item each payoff function is defined by
\[
p_i(s) = \sum_{j \in N_i,\, s_i = s_j} w_{j \to i} + \beta(i,s_i).
\]
\end{itemize}

So each node simultaneously chooses a colour and the payoff to the
node is the sum of the weights of the edges from its neighbours that
chose its colour augmented by the bonus to the node from choosing the
colour.  We call these games \bfe{coordination games on directed
  graphs}, from now on just \bfe{coordination games}.  Because weights
are non-negative each coordination game satisfies the PPM.

In the paper we mostly consider the case when all weights are 1 and all bonuses
are 0. Then each payoff function is simply defined by
\[
p_i(s) := |\{j \in N_i \mid s_i = s_j\}|.
\]

\begin{example} \label{exa:payoff}
\rm

Consider the directed graph and the colour assignment
depicted in Figure~\ref{fig:graph}.

\begin{figure}[htbp]
\centering
$
\def\objectstyle{\scriptstyle}
\def\labelstyle{\scriptstyle}
\xymatrix@R=25pt @C=40pt{
& & 7 \ar[d]^<{\{\underline{a}\}}\\
& &1 \ar[rd]_{} \ar[ddr]_{}\ar@{}[rd]^<{\{a,\underline{b}\}}\\
&6 \ar[ru]_{} \ar@{}[ur]^<{\{b,\underline{c}\}}& &4 \ar[d]_{}\ar@{}[ul]_<{\{a,\underline{b}\}} \\
&3 \ar[u]^<{\{b,\underline{c}\}} \ar[uur]_{}& &2 \ar@{}[u]_<{\{a,\underline{c}\}} \ar[ld]^{} \ar[ll]_{}\\
9 \ar[ru]_<{\{\underline{b}\}}& &5 \ar[ul]_{}\ar@{}[ur]_<{\{a,\underline{c}\}}& &8 \ar[ul]^<{\{\underline{c}\}}
}$
    
\caption{A directed graph with a colour assignment. \label{fig:graph}}
\end{figure}
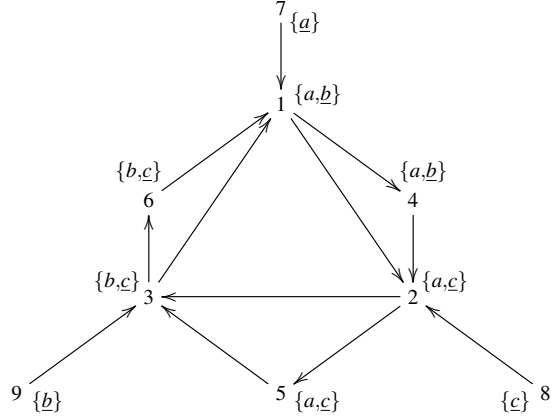

Take the joint strategy $s$ that consists of the underlined strategies.
% , so
%
% \begin{itemize}
%
% \item node 7 selects $a$,
%
% \item nodes 1, 4 and 9 select $b$,
%
% \item nodes 2, 3, 5, 6 and 8 select $c$.
%
% \end{itemize}
Then the payoffs are as follows:

\begin{itemize}
\item 0 for the nodes 1, 7, 8 and 9,

\item 1 for the nodes 2, 4, 5, 6, 

\item 2 for the node 3.
\end{itemize}
Note that the above joint strategy is not a Nash equilibrium. For example,
node 1 can profitably deviate to colour $a$.
\HB
\end{example}

In what follows we study the problem of existence of
Nash equilibria or strong equilibria in coordination games.

Finally, given a directed graph $G$ and a set of nodes $K$, we denote by $G[K]$
the subgraph of $G$ induced by $K$.

% \marginpar{NOT USED?}

%  and by $E[K]$ the set of edges that
% have both endpoints in $K$. So $G[K] = (K, E[K])$.  Further,
% $\delta(K)$ denotes the set of edges one vertex of which is in $K$ and
% the other outside of $K$.

% Consider a coordination game $(S_1, \ldots, S_n, p_1, \ldots, p_n)$ on
% a graph $(V,E)$ and a set of nodes $K$.  We define $SW_K(s) := \sum_{i
%   \in K} p_i(s)$. Further, given a joint strategy $s$ we denote by
% $E^+_s$ the set of edges $(i,j) \in E$ such that $s_i = s_j$. We call
% these edges \bfe{unicolour (in $s$).} Note that $SW(s) = 2|E_s^+|$.
% Finally, we call a subgraph \bfe{unicolour in $s$} if all its nodes
% have the same colour in $s$.

% \marginpar{END USED?}

\section{Strong equilibria}
\label{sec:directed-strong}

In this section we focus on the existence of strong equilibria.
To start with, we have the following positive result.

\begin{theorem} \label{thm:DAG}
Every coordination game whose underlying graph is a DAG
has the c-FIP and a fortiori a strong equilibrium.
Further, each Nash equilibrium is a strong equilibrium.
\end{theorem}

\begin{proof}
Given a DAG $G := (V,E)$, where $V = \{1, \LL, n\}$, we fix a 
permutation $\pi$ of $1, \LL, n$ such that for all $i,j \in V$ 
\begin{equation}
 \label{equ:rank}
\mbox{if $i < j$, then $(\pi(j) \to \pi(i)) \not \in E$.}
%no path in $G$ from $j$ to $i$.}
\end{equation}
So if $i < j$, then the payoff of the node $\pi(i)$ does not depend
on the strategy selected by the node $\pi(j)$.

Then given a coordination game whose underlying directed graph is the
DAG $G$ we associate with each joint strategy $s$ the sequence
$p_{\pi(1)}(s), \LL, p_{\pi(n)}(s)$ that we abbreviate to
$p_{\pi}(s)$.  We now claim that $p_{\pi}: S \to \mathbb{R}^n$ is a
generalized ordinal c-potential when we take for the partial ordering
$\succ$ on $p_{\pi}(S)$ the lexicographic ordering $>_{lex}$ on the
sequences of reals.

Suppose that some coalition $K$ profitably deviates from the joint
strategy $s$ to $s'$.  Choose the smallest $j$ such that $\pi(j) \in K$.
Then $p_{\pi(j)}(s') > p_{\pi(j)}(s)$ and by (\ref{equ:rank})
$p_{\pi(i)}(s') = p_{\pi(i)}(s)$ for $i < j$.  This implies that
$p_{\pi}(s') >_{lex} p_{\pi}(s)$, as desired.
Hence the game has the c-FIP.

To prove the second claim, take a Nash equilibrium $s$ and suppose it
is not a strong equilibrium. Then some coalition $K$ can profitably
deviate from $s$ to $s'$.  Choose the smallest $j$ such that $\pi(j)
\in K$.  Then $p_{\pi(j)}(s') > p_{\pi(j)}(s)$ and by (\ref{equ:rank})
the payoff of $\pi(j)$ does not depend on the strategies selected by
the other members of the coalition $K$. Hence
$p_{\pi(j)}(s') = p_{\pi(j)}(s'_{\pi(j)}, s_{-\pi(j)})$, which contradicts
the assumption that $s$ is a Nash equilibrium.
\end{proof}

The next result deals with a class of coordination games introduced in
\cite{ARSS15}.  Given the set of colours $M$, we say that a directed
graph $G$ is \bfe{colour complete} (with respect to a colour
assignment $A$) if for every colour $x \in M$ each component of
$G[V_x]$ is complete, where $V_x = \{i \in V \mid x \in A_i\}$. In
particular, every complete graph is colour complete.

\begin{theorem} \label{thm:complete2}
  Every coordination game on a colour complete directed graph has the c-FIP
  and a fortiori a strong equilibrium.
\end{theorem}

\begin{proof}
  In \cite{ARSS15} it is proved that every uniform game has
  the c-FIP, where we call a coordination game on a directed graph $G$
  \bfe{uniform} if for every joint strategy $s$ and for every edge $i
  \to j \in E$ it holds: if $s_i = s_j$ then $p_i(s) = p_j(s)$.  (In
  \cite{ARSS15} only undirected graphs are considered, but the proof
  remains valid without any change.)  Clearly every coordination game
  on a colour complete directed graph is uniform.
\end{proof}

It is difficult to come up with other classes of directed
graphs for which the coordination game has the c-FIP. Indeed,
consider the following example.

\begin{example} \label{exa:rotate}
\rm

Consider a coordination game on a simple cycle $1 \to 2 \to \LL \to n
\to 1$, where $n \geq 3$ and such that the nodes share at least two
colours, say $a$ and $b$.  Take the initial colouring $(a, b, \LL,
b)$. Then both $(a, \underline{b}, b, \LL, b), (a, a, b, \LL, b)$ and
$(\underline{a}, a, b, \LL, b), (b, a, b, \LL, b)$ are profitable
deviations. (To increase readability we underlined the strategies that
were modified.)  After these two steps we obtain a colouring that is a
rotation of the first one. Iterating we obtain an infinite improvement
path.

Hence the coordination game does not have the FIP and a fortiori
the c-FIP.
\HB
\end{example}

However, a weaker result holds, which, for reasons that will soon
become clear, we prove for a larger class of games.

\begin{theorem} \label{thm:c-weakly}
  Every coordination game with bonuses on a simple
  cycle is c-weakly acyclic, so a fortiori has a strong equilibrium.
\end{theorem}

To prove it, we first establish a weaker claim.

\begin{lemma} \label{lem:weakly}
  Every coordination game with bonuses on a simple
  cycle is weakly acyclic.
\end{lemma}
\begin{proof}
  To fix the notation, suppose that the considered graph is $1 \to 2
  \to \ldots \to n \to 1$.  Below for $i \in
\C{2,\LL,n}$, $i \ominus 1=i-1$, and $1 \ominus 1=n$.

Let $\mathit{MA}(i)$ be the set of available colours to player $i$
with the maximal bonus, i.e., $\mathit{MA}(i) = \{c \in A(i) \mid
\beta(i,c) = \max_{d \in A(i)} \beta(i,d)\}$.  Let
%\[
% \mathit{BR}(i,\strprofile_{-i})=\{c \in \mathit{MA}(i) \mid
% \text{colour $c$ is a best response of player $i$ to
%   $\strprofile_{-i}\}$}
%\] 
%%\begin{tabbing}
$\mathit{BR}(i,\strprofile_{-i})=\{c \in \mathit{MA}(i) \mid \text{colour $c$ is a best response}$ $\text{of player $i$ to  $\strprofile_{-i}$\}}$
%%\end{tabbing}
be the set of best responses among the colours with the highest bonus
only.  The set $\mathit{BR}(i,\strprofile_{-i})$ is never empty
because of the game structure and the fact that bonuses are natural
numbers.  Indeed, if $\strprofile_{i \ominus 1} \in \mathit{MA}(i)$,
then $\mathit{BR}(i,\strprofile_{-i}) = \{\strprofile_{i \ominus 1}\}$
and otherwise $\mathit{BR}(i,\strprofile_{-i})$ is a non-empty
subset of $\mathit{MA}(i)$.
  
  Below we stipulate that whenever a player $i$ updates in a joint
  strategy $\strprofile$ his strategy to a best response to $s_{-i}$,
  he always selects a strategy from $\mathit{BR}(i,\strprofile_{-i})$.
  
  Consider an initial joint strategy $\strprofile$.  We construct
  a finite improvement path that starts with $\strprofile$ as follows.\\
  {\em Phase 1.} We proceed around the cycle and consider the players
  $1$, $2$, \ldots, $n-1$ in that order.  For each player in turn, if
  his current strategy is not a best response, we update it to a
  best response respecting the above proviso.  When this phase ends
  the current strategy of each of the players $1$, $2$, \ldots, $n-1$
  is a best response. 
  
  If at this moment the current strategy of player $n$ is also a best
  response, the current joint strategy $s'$ is a Nash equilibrium and
  the path is constructed. Otherwise we move to the next phase.  \\
  {\em Phase 2.} We repeat the same process as in Phase~1, but
  starting with $s'$ and player $n$ and proceeding at most $n$ steps.
  From now on at each step at least $n-1$ players have a best response
  strategy. So if at a certain moment the current strategy of the
  considered player is a best response, the current joint strategy is
  a Nash equilibrium and the path is constructed.
  Otherwise, after $n$ steps, we move to the final phase.  \\
  {\em Phase 3.} We repeat the same process as in Phase~2. Now in the
  initial joint strategy each player $i$ has a strategy from
  $\mathit{MA}(i)$.  Because of the definition of
  $\mathit{BR}(i,\strprofile_{-i})$ each player can improve his payoff
  only if he switches to the strategy selected by his predecessor.  So
  after at most $n$ steps this phase terminates and we obtain a Nash
  equilibrium.
\end{proof}

By Lemma \ref{lem:weakly} every coordination
game on a simple cycle has a Nash equilibrium. However, not every
Nash equilibrium is then a strong equilibrium.

\begin{example}
\rm
Consider the directed graph depicted in Figure~\ref{fig:2}, together with the sets of colours
associated with the nodes. 

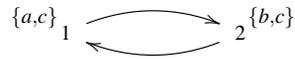
\begin{figure}[htbp]
\centering
$ \def\objectstyle{\scriptstyle} \def\labelstyle{\scriptstyle}                         
  \xymatrix@W=10pt @C=50pt @R=5pt{
& \\
1 \ar@/^0.7pc/[r]^{} \ar@{}[u]^<{\{a,c\}}& 2 \ar@/^0.7pc/[l]^{}
\ar@{}[u]_<{\{b,c\}}                        
}$
  \caption{Nash equilibria versus strong equilibria}
  \label{fig:2}
\end{figure}

Clearly $(a,b)$ is a Nash equilibrium. However, it is not a strong
equilibrium since the coalition $\{1,2\}$ can profitably deviate
to $(c,c)$, which is a strong equilibrium.
\HB
\end{example}

On the other hand, the following observation holds.

\begin{lemma} \label{lem:k}
  Consider a coordination game with bonuses on a simple
  cycle with $n$ nodes. Then every Nash equilibrium is a
  $k$-equilibrium for all $k \in \{1, \LL, n-1\}$.
\end{lemma}

\begin{proof}
Take a Nash equilibrium $s$. It suffices to prove that it is an $(n-1)$-equilibrium.
Suppose otherwise. Then for some coalition $K$ of size $\leq n-1$ and a joint strategy $s'$,
$s \betredge{K} s'$ is a profitable deviation. 

Assume $k \ominus 1 = k-1$ if $k > 1$ and $1 \ominus 1 = n$.
Take some $i \in K$ such that $i \ominus 1 \not \in K$.  We have
$p_i(s') > p_i(s)$. Also $p_i(s'_i, s_{-i}) = p_i(s')$, since $s_{i
  \ominus 1} = s'_{i \ominus 1}$.  So $p_i(s'_i, s_{-i}) > p_i(s)$,
which contradicts the fact that $s$ is a Nash equilibrium.
\end{proof}

\NI
\emph{Proof of Theorem \ref{thm:c-weakly}.}  Take a joint strategy
$s$. By Lemma~\ref{lem:weakly} a finite improvement path starts
at $\strprofile$ and ends in a Nash equilibrium $s'$.  By
Lemma~\ref{lem:k} $s'$ is an $(n-1)$-equilibrium. If $s'$ is not a
strong equilibrium, then a profitable deviation $s' \betredge{N} s''$
exists, where, recall, $N$ is the set of all players.  Because of the game
structure the social welfare along each c-improvement path weakly
increases, while in the last step the social welfare strictly
increases.  So $\mathit{\SW}(s'') > \mathit{\SW}(s)$.

If $s''$ is not a strong equilibrium, we repeat the above procedure starting with $s''$.
Since each time the social welfare strictly increases, eventually 
this process stops and we obtain a finite c-improvement path.
\HB
\VV

%% The reason that we formulated Theorem \ref{thm:c-weakly} for games with bonuses is that now
%% we can prove the following stronger result.

Using Theorem \ref{thm:c-weakly}, we now show that every coordination
game in which all strongly connected components are simple cycles is
c-weakly acyclic.  We first introduce some notations and make use of
the following well-known decomposition result.

\begin{theorem}[\cite{DPV06}, page 92]
Every directed graph is a directed acyclic graph of its strongly
connected components.
%%(SCCs)
\end{theorem}

Given a graph $G=(V,E)$, let $D=(V_D,E_D)$ be the corresponding
directed acyclic graph (DAG) obtained by the above decomposition
theorem and let $g:2^V \to V_D$ be the function that maps each
strongly connected component (SCC) in $G$ to a node in $D$. Let
$g^{-1}(v)=X \subseteq V$ where $g(X)=v$. Note that for each $i \in
V$, there is a unique $v \in V_D$ such that $i \in g^{-1}(v)$, we
denote this node by $v_i$.  Let $|V_D|=m$ and
$\theta=(\theta_1,\theta_2,\ldots,\theta_m)$ be a topological ordering
of $V_D$ (this is well-defined since $D$ is a DAG). We define a
labelling function $l_D:V_D \to \{1,\ldots,m\}$ as follows: for all $v
\in V_D$, $l_D(v)=j$ iff $\theta_j = v$. We can extend $l_D$ to a
function $l:V \to \{1,\ldots,m\}$ in the natural way: for all $i \in
V$, $l(i)=l_D(v)$ if $i \in g^{-1}(v)$.

Note that for each node $v \in V_D$, either $g^{-1}(v)=\{i\}$ for some
$i \in V$ or $g^{-1}(v)=X \subseteq V$ with $|X| \geq 2$ and the
subgraph of $G$ induced by the set of nodes $X$ forms an SCC. Also,
note that every $v \in V_D$ and a joint strategy $\strprofile$ in
$\calg$, defines a coordination game with bonuses $\calg_v$ on the
graph $G(v,\strprofile) = (V',E')$ which is the subgraph induced by
the set of nodes $V'=g^{-1}(v)$. For $i \in V'$ and $a \in A(i)$ we
put $\beta(i,a) := |\{j \in N_i \setminus V' \mid s_j=a\}|$. 
%% Given two
%% joint strategies $s$ and $t$ in $\calg$, we define
%% $\mathit{deviate}(s,t)=\{i \in V \mid s_i \neq t_i\}$.

\begin{theorem} \label{thm:SCCs}
Every coordination game on a directed graph $G$ in which all strongly
connected components of $G$ are simple cycles is c-weakly acyclic and
a fortiori has a strong equilibrium.
\end{theorem}

\begin{proof}
Consider a coordination game $\calg$ on a graph $G=(V,E)$ where all
SCCs are simple cycles. Let $D=(V_D,E_D)$ be the corresponding DAG
with $|V_D|=m$. Since all SCCs in $G$ are simple cycles, it follows
that for all $v \in V_D$, either $g^{-1}(v) = \{i\}$ or $g^{-1}(v)=X
\subseteq V$ such that the induced graph on $X$ forms a simple cycle in
$G$. 

Let $v \in V_D$ such that the induced graph on $g^{-1}(v)$ forms a
simple cycle in $G$. For a joint strategy $t$ in $\calg$, consider the
resulting game $\calg_v$ on the graph $(V',E')$.
%% and a
%% joint strategy $t$ in $\calg$.
Let $\strprofile^0=t_{V'}$ (the restriction of the joint strategy $t$
to nodes in $V'$) and let $\rho: \strprofile^0, \strprofile^1, \ldots,
\strprofile^k$ be a finite c-improvement path in $\calg_v$ which is
guaranteed to exist by Theorem \ref{thm:c-weakly}.
%% Given a coordination game with bonuses $\calg'$ whose underlying graph
%% is a simple cycle, and a joint strategy $\strprofile$ in $\calg'$, let
%% $\rho_s: \strprofile=\strprofile^0, \strprofile^1, \ldots,
%% \strprofile^k$ be a finite c-improvement path which is guaranteed to
%% exist by Theorem \ref{thm:c-weakly}.
Define $\cpath(\calg_v,t)$ as follows:

\NI
$\cpath(\calg_v,t) = 
\begin{cases} 
  \epsilon ~~~~\mbox{ if $t_{V'}$ is a strong equilibrium in $\calg_v$,} \\
  \lambda_t(\strprofile^1), \ldots, \lambda_t(\strprofile^k) ~~~~\mbox{ otherwise, }
\end{cases}
$

\noindent where for all $h \in \{1,\ldots,k\}$, $\lambda_t(\strprofile^h)$ is the
joint strategy in $\calg$ defined as: for all $i \in V$,
$(\lambda_t(\strprofile^h))_i=\strprofile^h_i$ if $i \in V'$ and
$(\lambda_t(\strprofile^h))_i= t_i$ if $i \not\in V'$.  

For a joint strategy $t$ in $\calg$ and $v \in V_D$, if the underlying
graph of the coordination game $\calg_v$ with bonuses consists of
exactly one node, then the game is trivially c-weakly
acyclic. $\cpath(\calg_v,t)$ is then defined analogously.

%% and we define $\cpath(\calg_v,t)$ analogously.

%% Given a joint strategy $\strprofile$ and a node $v \in V_D$, we denote
%% the restriction of $s$ to components of $g^{-1}(v)$ as
%% $\strprofile_{g^{-1}(v)}$.

%% For a sequence of strategy profiles
%% $\rho=\strprofile_0,\strprofile_1,\ldots,\strprofile_k$, let
%% $\lasts(\rho)=\strprofile_k$. 

Let $t^0$ be an arbitrary joint strategy in $\calg$.  We
define a sequence of joint strategies as follows:
\begin{itemize}
\item $\rho_0=t^0$,
\item for $h \in \{0,1,\ldots,m-1\}$, let $\rho_{h+1}=\rho_h \cdot
  \cpath(\calg_v,t^h)$ where $l_D(v)=h+1$ and $t^h=\lasts(\rho_h)$.
\end{itemize}

Let $\rho=\rho_m$. From the definition of $\rho_m$ and $\cpath$, it
follows that $\rho$ is finite sequence of joint strategies in
$\calg$. By induction on the length of $\rho$, we can claim that for
every subsequent pair of joint strategies $t^k$ and $t^{k+1}$ in
$\rho$, there is a coalition $K \subseteq V$ for which $t^{k+1}$ is a
profitable deviation from $t^k$. To complete the proof, it suffices to
argue that $\rho$ is maximal, or equivalently, that $\lasts(\rho)$ is
a strong equilibrium.

Suppose $\lasts(\rho)$ is not a strong equilibrium. Then there exists
$K \subseteq V$ and a joint strategy $\strprofile$ such that there is
a profitable deviation of players in $K$ from $\lasts(\rho)$ to
$\strprofile$. Let $d$ be the least element of the set $\{l(i) \mid i
\in K\}$ and $X=K \cap \{i \in V \mid l(i)=d\}$. By definition of a
profitable deviation, we have that for all $i \in X$, $p_i(s) >
p_i(\lasts(\rho))$. Note that for all $i \in X$ and for all $j \in N_i
\setminus g^{-1}(v_i)$, we have $l(j) < d$. Therefore, $(N_i
\setminus g^{-1}(v_i)) \cap K=\emptyset$. Also note that for all $j
\in g^{-1}(v_i)$, $(\lasts(\rho_d))_{j}=(\lasts(\rho))_{j}$. But this
implies that the coalition $X$ has a profitable deviation from the
joint strategy $(\lasts(\rho_d))_X$ to $s_X$ in the game
$\calg_{v_i}$. This contradicts the fact that $\lasts(\rho_d)$ is a
strong equilibrium in the game $\calg_{v_i}$.
\end{proof}

We conclude this section by considering another class of coordination
games.  Example \ref{exa:rotate} shows that even when only two colours
are used, the coordination game does not need to have the c-FIP.  On
the other hand, a weaker property does hold.

\begin{theorem} \label{thm:two-c}
Every coordination game in which only two colours are used is c-weakly acyclic
and a fortiori has a strong equilibrium.
\end{theorem}

\begin{proof}
  We prove the result for a more general class of games, namely the
  ones that satisfy the PPM.  Call the colours blue and red, that we
  abbreviate to $b$ and $r$. When a node selected blue we refer to it
  as a blue node, and the same for the red colour.

Take a joint strategy $s^1$. Consider a maximal sequence $\xi$ of
profitable deviations of the coalitions starting in $s$ in which the
nodes can only switch to blue. At each step the number of blue nodes
increases, so $\xi$ is finite. Let $s^1, \LL, s^k$, where $k \geq 1$,
be the successive joint strategies of $\xi$.

If $s^k$ is a strong equilibrium, then $\xi$ is the desired finite
improvement path. Otherwise consider a maximal sequence $\chi$ of
profitable deviations of the coalitions starting in $s^k$ in which the
nodes can only switch to red. $\chi$ is finite.  Let $s^k, s^{k+1},
\LL, s^{k+l}$, where $l \geq 1$, be the successive joint strategies of
$\chi$.

We claim that $s^{k+l}$ is a strong equilibrium. Suppose otherwise.
Then for some joint strategy $s'$, $s^{k+l} \betredge{K} s'$ is a
profitable deviation of some coalition $K$. Let $L$ be the set of
nodes from $K$ that switched in this deviation to blue.  By the
definition of $s^{k+l}$ the set $L$ is non-empty.

Given a set of nodes $M$ and a joint strategy $s$ we denote by
$(M:b, s_{-M})$ the joint strategy obtained from $s$ by letting the
nodes in $M$ to select blue, and similarly for the red colour.  Also
it should be clear what joint strategy we denote by $(M:b, P
\setminus M:r, s_{-P})$, where $M \sse P$.
 
We claim that $s^{k+l} \betredge{L} (L:b, s^{k+l}_{-L})$ is a profitable deviation of
the players in $L$. Indeed, we have for all $i \in L$
\begin{equation}
p_i(L:b, s^{k+l}_{-L}) > p_i(s^{k+l}), 
  \label{equ:k+l}
\end{equation}
since by the PPM $p_i(L:b, s^{k+l}_{-L}) \geq
p_i(s')$ and by the assumption $p_i(s') > p_i(s^{k+l})$.

Let $M$ be the set of nodes from $L$ that are red in $s^k$.
Suppose that $M$ is non-empty. We show that then
\begin{equation}
  \label{equ:M}
  p_M(M:r, L \setminus M:b, s^k_{-L}) < p_M(M:b, L \setminus M:b, s^k_{-L}).
\end{equation}
Indeed, we have for all $i \in M$
\[
\begin{array}{ll}
     & p_i(M:r, L \setminus M:b, s^{k}_{-L}) \leq p_i(M:r, L \setminus M:b, s^{k+l}_{-L}) \\
\leq & p_i(M:r, L \setminus M:r, s^{k+l}_{-L}) < p_i(M:b, L \setminus M:b, s^{k+l}_{-L}) \\
\leq & p_i(M:b, L \setminus M:b, s^{k}_{-L}),
\end{array}
\]
where the weak inequalities are due to the PPM and the
strict inequality holds by the definition of $L$.

But $s^k = (M:r, L \setminus M:b, s^{k}_{-L})$, so
(\ref{equ:M}) contradicts the definition of $s^k$.
So $M$ is empty, i.e., all nodes from $L$ are blue in $s^k$.
We now have for all $i \in L$
\[
\begin{array}{ll}
  & p_i(L:r, s^k_{-L}) \leq p_i(L:r, s^{k+l}_{-L}) = p_i(s^{k+l}) \\
< & p_i(L:b, s^{k+l}_{-L})  \leq p_i(L:b, s^{k}_{-L}),
\end{array}
\]
where again the weak inequalities are due to the PPM and the
strict inequality holds by (\ref{equ:k+l}).

But $(L:r, s^k_{-L})= s^{k}$, so we proved that $s^{k} \betredge{L}
(L:b, s^{k}_{-L})$ is a profitable deviation. This yields a
contradiction with the definition of $s^k$.
\end{proof}

When the underlying graph is symmetric and the set of strategies for
every node is the same, the existence of strong equilibrium for
coordination games with two colours follows from Proposition 2.2 in
\cite{KBW97a}. Theorem~\ref{thm:two-c} shows a stronger result, that
in the general case, these games are c-weakly acyclic. The following
example shows that when three colours are used, Nash equilibria, so a
fortiori strong equilibria do not need to exist.

\begin{example}\label{exa:no}
{\rm 
Consider the directed graph depicted in Figure~\ref{fig:graph} of
Example \ref{exa:payoff}, together with the sets of colours associated
with the nodes.
We argue that the coordination game associated with this graph does not
have a Nash equilibrium. Note that for nodes 7, 8 and 9 the only
option is to select the unique strategy in its strategy set. The best
response for nodes 4, 5 and 6 is to always select the same strategy as
nodes 1, 2 and 3 respectively. Therefore, to show that the game does
not have a Nash equilibrium, it suffices to consider the strategies of
nodes 1, 2 and 3. We denote this by the triple
$(\strprofile_1,\strprofile_2,\strprofile_3)$. Below we list all such
joint strategies and we underline a strategy that is not a best
response to the choice of other players: $(\underline{a},a,b)$,
$(a,a,\underline{c})$, $(a,c,\underline{b})$,
$(a,\underline{c},c)$, $(b,\underline{a},b)$,
$(\underline{b},a,c)$, $(b,c,\underline{b})$ and
$(\underline{b},c,c)$.  
}
\HB
\end{example}

Call now a graph a \bfe{coloured DAG} (with respect to a colour
assignment $A$) if for each available colour $x$ the components of the
subgraph induced by the nodes having colour $x$ are DAGs.  In view of
Theorem \ref{thm:complete2} it is tempting to try to generalize
Theorem \ref{thm:DAG} to coloured DAGs.  However, the directed graph
depicted in Figure~\ref{fig:graph} is a coloured DAG and, as explained
in the above example, the coordination game on this graph has no Nash
equilibrium.

\section{Complexity issues}
\label{sec:directed-complexity}

Next, we study the complexity of the existence problems and of the problem
of finding strong equilibria. 
%We first establish the following result.

\begin{figure}[hbp]
\centering
$
\def\objectstyle{\scriptstyle}
\def\labelstyle{\scriptstyle}
\xymatrix@R=25pt @C=40pt{
& & \bullet \ar[d]^<{\{{R}\}}_{2}\\
& &A_i \ar[rd]^{1} \ar[ddr]_{1}\ar@{}[rd]^<{\{R,{G},x\}}\\
& \bullet \ar[ru]_{} \ar@{}[ur]^<{\{G,{B}\}}^{2}& &\bullet
\ar[d]^{2}\ar@{}[ul]_<{\{R,{G}\}} \\
&C_i \ar[u]^<{\{G,{B},z\}}^{1} \ar[uur]_{1}& &B_i
\ar@{}[u]_<{\{R,{B},y\}} \ar[ld]^{1} \ar[ll]_{1}\\
\bullet \ar[ru]_<{\{{G}\}}^{2}& &\bullet
\ar[ul]_{2}\ar@{}[ur]_<{\{R,{B}\}}& &
\bullet \ar[ul]^<{\{{B}\}}^{2}
}$

\caption{Gadget $D_i$ with three parameters $x,y,z \in \{\top,\bot\}$ \label{fig:gadget} and three distinguished nodes $A_i, B_i, C_i$.}
\end{figure}
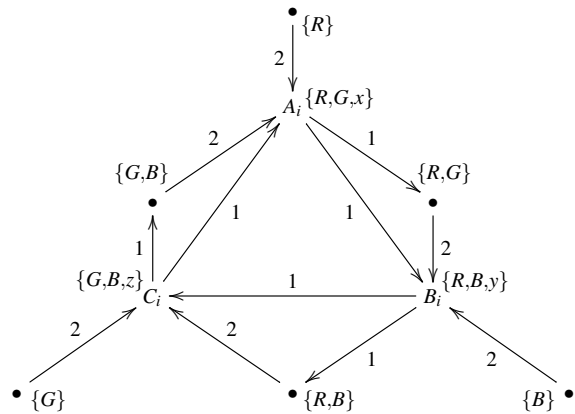

\begin{theorem} \label{thm:NE-NPcomplete}
The Nash equilibrium existence problem in coordination games is NP-complete.
\end{theorem}
\begin{proof}
  The problem is in NP, since we can simply guess a colour assignment
  and checking whether it is a Nash equilibrium can be done in polynomial time.
  
  To prove NP-hardness we provide a reduction from the 3-SAT problem,
  which is NP-complete.  Notice that an edge with a natural number
  weight $w$ can be simulated by adding $w$ extra players to the game.
  More precisely, an edge $(i \to j)$ with the weight $w$ can be
  simulated by the extra set of players $\{i_1, \ldots, i_w\}$ and the
  following $2\cdot w$ unweighted edges: $\{(i \to i_1), (i \to i_2),
  \ldots, (i \to i_w), (i_1 \to j), (i_2 \to j), \ldots, (i_w \to
  j)\}$. Given a colour assignment in the game with the weighted edges,
  we then assign to each of the nodes $i_1, \LL, i_w$ the colour set
  of the node $i$.
  
  Therefore we will assume that edges can have such weights assigned
  to them, because this simplifies our construction.  Assume we are
  given a 3-SAT formula
\[\phi = (a_1 \vee b_1 \vee c_1) \wedge (a_2 \vee b_2 \vee c_2) \wedge \ldots \wedge (a_k \vee b_k \vee c_k)\]
with $k$ clauses and $n$ propositional variables $x_1, \ldots, x_n$,
where each $a_i,b_i,c_i$ is a literal equal to $x_j$ or $\lnot x_j$
for some $j$. We will construct a coordination game $\game$ of size
$\mathcal{O}(k)$ such that $\game$ has a Nash equilibrium iff $\phi$
is satisfiable.

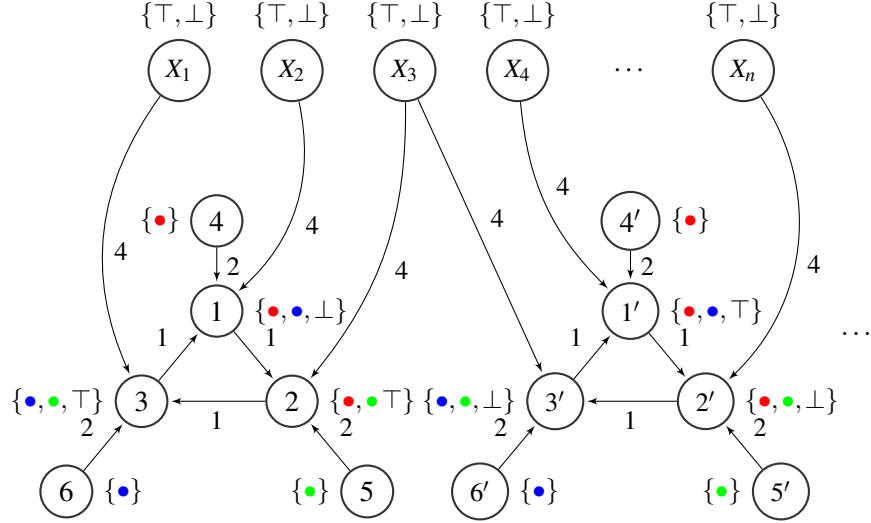
\begin{figure*}%[ht]
\centering
\newcommand{\redp}[1]{\textcolor{red}{#1}}
\newcommand{\bluep}[1]{\textcolor{blue}{#1}}
\newcommand{\greenp}[1]{\textcolor{green}{#1}}
\centering
%{
\tikzstyle{agent}=[circle,draw=black!80,thick, minimum size=1em]
\tikzstyle{changeblue}=[draw=blue!80,fill=blue!30,thick]
\tikzstyle{changered}=[draw=red!80,fill=red!30,thick]
\tikzstyle{changebrown}=[draw=brown!80,fill=brown!40,thick]
%\resizebox{!}{6cm}{
\begin{tikzpicture}[auto,>=latex',shorten >=1pt,on grid]
\begin{scope}
%\onslide<3->{
\node[agent,label=above:{\small $\{\top, \bot\}$}](z1){\small{$X_1$}};
\node[agent, right of=z1, node distance=1.5cm,label=above:{\small $\{\top, \bot\}$}](z2){\small $X_2$};
\node[agent, right of=z2, node distance=1.5cm,label=above:{\small $\{\top, \bot\}$}](z3){\small $X_3$};
\node[agent, right of=z3, node distance=1.5cm,label=above:{\small $\{\top, \bot\}$}](z4){\small $X_4$};
\node[draw=none,fill=none, right of=z4,node distance=1.5cm](z5){$\cdots$};
\node[agent, right of=z5, node distance=1.5cm,label=above:{\small $\{\top, \bot\}$}](z6){\small $X_n$};
%}
\end{scope}
%\begin{scope}[on background layer]
%\onslide<4->{
%\draw[->](z1) to node[swap]{{\small $6$}} (-2.3,-1.8);
%\draw (5) to [bend left=20] node {$p$} (6);
%\draw[->](z1) to node[swap]{{\small $6$}} to (a1);
%\draw[->,bend right,bend angle=2](z1) to node[swap] {{\small $a_1$}} (6.1,-2);
%}
%\onslide<5->{
%\draw[->](z2) to node{{\small $6$}} (-2.2,-1.9);
%\draw[->,bend right,bend angle=2](z2) to node {{\small $a_2$}} (6.1,-1.9);
%}
%\onslide<6->{
%\draw[->,bend left,bend angle=2](z4) to node[swap] {{\small $a_n$}} (-2.2,-2);
%\draw[->](z4) to node{{\small $6$}} (6.3,-1.7);
%}
%\end{scope}
\begin{scope}[xshift=0.5cm,yshift=-2cm]
\node[agent,label=left:{\small $\{\redp{\bullet}\}$}](a1) {$4$};
\node[agent,below of=a1, node distance=1.2cm,label=right:{\small $\{\redp{\bullet},\bluep{\bullet}, \bot\}$}](b1){$1$};
\node[below of=b1,node distance=1.2cm,coordinate] (cc){};
\node[agent,left of=cc, label=180:{\small  $\{\bluep{\bullet},\greenp{\bullet}, \top \}$}](c1){$3$};
\node[agent,right of=cc,label=0:{\small $\{\redp{\bullet},\greenp{\bullet}\, \top\}$}](c2){$2$};
\node[below of=c1,node distance=1.2cm,coordinate](dc1){};
\node[below of=c2,node distance=1.2cm,coordinate](dc2){};
\node[agent,left of=dc1,label=right:{\small $\{\bluep{\bullet}\}$}](d1){$6$};
\node[agent,right of=dc2,label=left:{\small $\{\greenp{\bullet}\}$}](d2){$5$};
\draw[->] (b1) to node {{\small $1$}} (c2);
\draw[->] (c2) to node {{\small $1$}} (c1);
\draw[->] (c1) to node {{\small $1$}} (b1);
\draw[->] (a1) to node {{\small $2$}} (b1);
\draw[->] (d1) to node {{\small $2$}} (c1);
\draw[->] (d2) to node [swap] {{\small $2$}} (c2);
\draw[->](z1) to [bend right=30] node{{\small $4$}} (c1);
\draw[->](z2) to [bend right=-30] node{{\small $4$}} (b1);
\draw[->](z3) to [bend right=-20] node{{\small $4$}} (c2);
\end{scope}

\begin{scope}[xshift=6cm,yshift=-2cm]
\node[agent,label=right:{\small $\{\redp{\bullet}\}$}](a1) {$4'$};
\node[agent,below of=a1, node distance=1.2cm,label=right:{\small $\{\redp{\bullet},\bluep{\bullet},\top\}$}](b1){\small{$1'$}};
\node[below of=b1,node distance=1.2cm,coordinate] (cc){};
\node[agent,left of=cc,label=180:{\small  $\{\bluep{\bullet},\greenp{\bullet},\bot\}$}](c1){\small{$3'$}};
\node[agent,right of=cc,label=0:{\small $\{\redp{\bullet},\greenp{\bullet},\bot\}$}](c2){\small{$2'$}};
\node[below of=c1,node distance=1.2cm,coordinate](dc1){};
\node[below of=c2,node distance=1.2cm,coordinate](dc2){};
\node[agent,left of=dc1,label=right:{\small $\{\bluep{\bullet}\}$}](d1){\small{$6'$}};
\node[agent,right of=dc2,label=left:{\small $\{\greenp{\bullet}\}$}](d2){\small{$5'$}};
\draw[->] (b1) to node {{\small $1$}} (c2);
\draw[->] (c2) to node {{\small $1$}} (c1);
\draw[->] (c1) to node {{\small $1$}} (b1);
\draw[->] (a1) to node {{\small $2$}} (b1);
\draw[->] (d1) to node {{\small $2$}} (c1);
\draw[->] (d2) to node [swap] {{\small $2$}} (c2);
\draw[->](z3) to node{{\small $4$}} (c1);
\draw[->](z4) to [bend right=20] node{{\small $4$}} (b1);
\draw[->](z6) to [bend right=-40] node{{\small $4$}} (c2);
\end{scope}

\begin{scope}[xshift=9cm,yshift=-3.5cm]
\node{$\ldots$};
\end{scope}

\end{tikzpicture}
%}
%}
\caption{\label{fig:exa}The game $\game$ corresponding to the formula $\phi = (x_1 \vee \lnot x_2 \vee x_3)   \wedge (\lnot x_3 \vee x_4 \vee \lnot x_n)$, where in each gadget the nodes of indegree 1 are omitted.}
\end{figure*}

First, for every propositional variable $x_i$ we have a corresponding
node $X_i$ in $\game$ with two possible colours $\top$ and $\bot$.
Intuitively, for a given truth assignment, if $x_i$ is true then
$\top$ should be chosen for $X_i$ and otherwise $\bot$ should be
chosen.  In our construction we make use of the following gadget,
denoted by $D_i(x,y,z)$, with three parameters $x,y, z \in \{\top,
\bot\}$ and $i$ used just for labelling purposes, and presented in
Figure \ref{fig:gadget}.  This gadget behaves similarly to the game
without Nash equilibrium analyzed in Example \ref{exa:no}.

What is important is that for all possible parameters values, the
gadget $D_i(x,y,z)$ does not have a Nash equilibrium. Indeed, each of
the nodes $A_i$, $B_i$, or $C_i$ can always secure a payoff 2, so
selecting $\top$ or $\bot$ is never a best response and hence in no
Nash equilibrium a node chooses $\top$ or $\bot$. The rest of the
reasoning is as in Example \ref{exa:no}.

For any literal, $l$, let 
\[
\isPos(l) := \begin{cases}
 \top & \mbox{if $l$ is a positive literal} \\
 \bot & \mathrm{otherwise}
\end{cases}
\]

For every clause $(a_i \vee b_i \vee c_i)$ in $\phi$ we
add to the game graph $\game$ the $D_i(\isPos(a_i),
\isPos(b_i), \isPos(c_i))$ instance of the gadget.
Finally, for every literal $a_i$, $b_i$, or $c_i$ in $\phi$, which is
equal to $x_j$ or $\lnot x_j$ for some $j$, we add an edge from $X_j$
to $A_i$, $B_i$, or $C_i$, respectively, with weight $4$.
We depict an example game $\game$ in Figure \ref{fig:exa}.

% Notice that, if allowed by the respective colour set, this guarantees
% that in each Nash equilibrium the colour of any of these three
% distinguished nodes $A_i$, $B_i$ and $C_i$ will be the same as the
% colour of its corresponding $X_j$.

% Indeed, the other weights are always at most 3.  

We claim that $\game$ has a Nash equilibrium iff $\phi$ is
satisfiable.
 
\smallskip \noindent ($\Rightarrow$) Assume there is a Nash
equilibrium $s$ in the game $\game$.  We claim that the truth
assignment $\nu :\{x_1, \ldots, x_n\} \to \{\top,\bot\}$ that assigns
to each $x_j$ the colour selected by the node $X_j$ in $s$ makes
$\phi$ true.  Fix $i \in \{1, \LL, k\}$. We need to show that $\nu$
makes one of the literals $a_i$, $b_i$, $c_i$ of the
clause $(a_i \vee b_i \vee c_i)$ true.

From the above observation about the gadgets it follows that at least
one of the nodes $A_i, B_i$, $C_i$ selected in $s$ the same colour as
its neighbour $X_j$.  Without loss of generality suppose it is $A_i$.
The only colour these two nodes, $A_i$ and $X_j$, have in common is
$\isPos(a_i)$. So $X_j$ selected in $s$ $\isPos(a_i)$, which by the
definition of $\nu$ equals $\nu(x_j)$. Moreover, by construction $x_j$
is the variable of the literal $a_i$.  But $\nu(x_j) = \isPos(a_i)$
implies that $\nu$ makes $a_i$ true.

\smallskip \noindent ($\Leftarrow$) Assume $\phi$ is satisfiable. Take
a truth assignment $\nu :\{x_1, \ldots, x_n\} \to \{\top,\bot\}$ that
makes $\phi$ true.  For all $j$, we assign to the node $X_j$ the
colour $\nu(x_j)$. We claim that this assignment can be extended to a
Nash equilibrium in $\game$.

Fix $i \in \{1, \LL, k\}$ and consider the $D_i(\isPos(a_i),
\isPos(b_i), \isPos(c_i))$ instance of the gadget.  The truth
assignment $\nu$ makes the clause $(a_i \vee b_i \vee c_i)$ true.
Suppose without loss of generality that $\nu$ makes $a_i$ true. We
claim that then it is always a unique best response for the node $A_i$
to select the colour $\isPos(a_i)$.

Indeed, let $j$ be such that $a_i = x_j$ or $a_i = \lnot x_j$.  Notice
that the fact that $\nu$ makes $a_i$ true implies that $\nu(x_j) =
\isPos(a_i)$. So when node $A_i$ selects $\isPos(a_i)$, the colour
assigned to $X_j$, its payoff is 4.

This partial assignment of colours can be completed to a Nash
equilibrium. Indeed, remove from the directed graph of $\game$ all
$X_j$ nodes and the nodes that secured the payoff 4, together with the
edges that use any of these nodes. The resulting graph has no cycles,
so by Theorem \ref{thm:DAG} the corresponding coordination game has a
Nash equilibrium.  Combining both assignments of colours we obtain a
Nash equilibrium in $\game$.
%
% Further, $A_i$ has an incoming edge with weight $4$ from a node with
% its colour set to $\isPos(a_i)$.  At the same time, the maximum
% payoff for $A_i$ for choosing a regular colour, i.e $R, G$ or $B$, is
% at most $3$.  A similar reasoning works for $B_i$ and $C_i$.
%
% However, since $\nu$ makes $\phi$ true, for all $i=1,2,\ldots,k$ at
% least one of $\nu(a_i)$, $\nu(b_i)$ or $\nu(c_i)$ is $\top$.  This
% implies that the game $\game$ has a Nash equilibrium, because we
% simply select for each $A_i, B_i, C_i$ the colour of $\isPos(a_i)$
% if it is available or the colour of the outer node connected to it
% with weight $2$.  Notice that setting any of $A_i, B_i$ or $C_i$ to
% $\top$ or $\bot$ does not prevent the other two nodes in that cycle
% from choosing a different colour.
\end{proof}

\begin{corollary}
The strong equilibrium existence problem in coordination games is NP-complete.
\end{corollary}

\begin{proof}
  It suffices to note that in the above proof the ($\Rightarrow$)
  implication holds for a strong equilibrium, as well, while in the
  proof of the ($\Leftarrow$) implication by virtue of Theorem
  \ref{thm:DAG} actually a strong equilibrium is constructed.
\end{proof}

An interesting application of Theorem \ref{thm:NE-NPcomplete} is in
the context of polymatrix games.
%{\em polymatrix games} \cite{Jan68}.
% \emph{polymatrix games} \cite{Jan68,How72}. 
 These are finite strategic form games in which the influence of a pure strategy
 selected by any player on the payoff of any other player is always the
 same, regardless of what strategies other players select.  
 Formally,
% it is a game $(S_1,\ldots,S_n,\payoff_1,\ldots,\payoff_n)$ in which
 for all pairs of players $i$ and $j$ there exists a partial payoff
 function $a^{ij}$ such that for any joint strategy
 $\strprofile=(\strprofile_1,\ldots,\strprofile_n)$, the payoff of
 player $i$ is given by $\payoff_i(\strprofile):=\sum_{j \neq i}
 a^{ij}(\strprofile_i,\strprofile_j)$.
In \cite{SA15} we proved that deciding whether a polymatrix game has a
Nash equilibrium is NP-complete. We can strengthen this result as follows.

\begin{theorem}
\label{thm:polymatrix-NP}
Deciding whether a polymatrix game with 0/1 partial payoffs has a Nash equilibrium is NP-complete.
\end{theorem}

\begin{proof}
  We can efficiently translate any coordination game
  $\mathcal{G}(G,M,w,A,\beta)$ into a polymatrix game $\mathcal{P}$
  with only 0/1 partial payoff as follows.  The number of players in
  $\mathcal{P}$ will be equal to the number of nodes in $G$ and the
  set of strategies for each player will be $M$.  We define
  $a^{ij}(\strprofile_i,\strprofile_j):=1$ if $\strprofile_i =
  \strprofile_j$ and $j \in N_i$, and
  $a^{ij}(\strprofile_i,\strprofile_j):=0$ otherwise.  
  
  Notice that any joint strategy
  $\strprofile=(\strprofile_1,\ldots,\strprofile_n)$ in $\mathcal{G}$
  is also a joint strategy in $\mathcal{M}$ with exactly the same
  payoff, because $\payoff^\mathcal{P}_i(\strprofile)=\sum_{j \neq i}
  a^{ij}(\strprofile_i,\strprofile_j) = |\{j \in N_i \mid s_i = s_j\}|
  = p^\mathcal{G}_i(s)$.  It follows that Nash equilibria in
  $\mathcal{G}$ and $\mathcal{P}$ coincide.  In particular, there
  exists Nash equilibrium in $\mathcal{G}$ if and only if there exists
  one in $\mathcal{P}$, but the former problem was shown to be {\sc
    NP}-hard in Theorem \ref{thm:NE-NPcomplete}, so the latter is also
  {\sc NP}-hard. On the other hand, for any polymatrix game we can
  guess a joint strategy and check whether it is a Nash equilibrium in
  polynomial time, which shows this decision problem is in fact
  NP-complete.
\end{proof}

Next, we determine the complexity of finding a strong equilibrium.
%  when
% in the underlying directed graph all strongly connected components of
% $G$ are simple cycles. 
%By Theorem \ref{thm:SCCs} we know that it exists.  
We begin with the following auxiliary result.

\begin{theorem} \label{thm:linear-time-on-cycle}
A strong equilibrium of a coordination game with bonuses on a simple
cycle can be found in linear time.
\end{theorem}

\begin{proof}
Let $n$ be the number of players in the game and $C$ the number of
possible colours. We assume adjacency list representation for the game
graph, binary representation of the bonuses and that the list of
colours available to player $i$ is given as a list of length $|A(i)|$
of elements of size $\log C$. Formally, the size of the input for
player $i$ only is $\Theta(|A(i)|\log C + \sum_{c\in A(i)} \log
(\beta(i,c)+1))$; the sum of these over $i=1,\ldots,n$ gives the total
size of the input.  

First note that for any colour assignment, the best response of the
$i$-th player can be found in time linear in the size of her part of
the input just by checking all possible colours in $A(i)$. Second,
each phase of the algorithm in Lemma \ref{lem:weakly} looks for the best
response (with a preference given to colours with a higher bonus) of
each player at most once, which will require time linear in the size
of the whole input.  The algorithm requires at most three such phases
before a Nash equilibrium is found, so it runs in linear time.

Note that thanks to Lemma \ref{lem:k} we know that any NE in such a
game structure is already a $(n-1)$-equilibrium, so the only way this
joint strategy is not a strong equilibrium is when all $n$ players can strictly
improve their payoff.  However, in any Nash equilibrium a player has to have her payoff at
at most one below the maximum possible one, because that is the minimum payoff for picking a colour with
the highest bonus.
Moreover, player's payoff can only be a natural number.  

Therefore, the only possibility when a NE is not a strong equilibrium 
is when there is a joint strategy which gives all the players
their maximum possible payoff, i.e.~each player is assigned
a colour with the highest possible bonus as well as gets
an extra $+1$ to her payoff for colour agreement with her
only neighbour. 
The latter implies that all the players need to pick the same colour in such a joint strategy.

To check whether such a joint strategy exists we do the following. 
Let $p = \text{argmin}_i |A(i)|$ be the player with the least number
of colours to choose from. 
We pick the set of her colours with the maximal bonus and 
intersect it with the set of colours with the maximal bonus for every other player.
An intersection of two sets represented as
lists of length $a$ and $b$ of elements of size $K$ can be done in
$\Theta(aK+bK)$ time, so the total running time will be
$\Theta(n |A(p)| \log C+ \sum_{i=1}^n |A(i)| \log C) =
\Theta(\sum_{i=1}^n |A(i)| \log C)$, because $|A(p)| \leq |A(i)|$ for
all $i$, which is linear. 
% in the size of the input. 
If the final set is empty then any NE is a strong equilibrium and otherwise
we know how to construct one.
%
%. The last expression is also linear in the size of the input
%so the total running time, as well.
\end{proof}

\begin{corollary} 
\label{cor:linear-time-on-scc}
A strong equilibrium of a coordination game on a graph in which all
strongly connected components are simple cycles can be computed
in linear time.
%\HB
\end{corollary}

\section{Conclusions}
\label{sec:conc}

We presented here a study of a simple class of coordination games on
directed graphs. We focused on the existence of Nash and strong
equilibria. We also studied the complexity of checking for the
existence of Nash and strong equilibria, as well as the complexity of
computing a strong equilibrium in certain cases where it is guaranteed
to exist.

A number of open problems remain. We showed that in general Nash
equilibria and strong equilibria are not guaranteed to exist.
However, if the underlying graph is a DAG, is colour
complete or is such that every SCC is a simple cycle, then strong
equilibria always exist. It would be interesting to identify other
classes of graphs for which Nash or strong equilibria exist.

The proof of Lemma \ref{lem:weakly} shows that in the case of a simple
cycle, starting from any initial joint strategy a Nash equilibrium can
be found by an improvement path of length at most $3n$. Also, each
step of such a path can be constructed in linear time.  Additionally,
the proof of Theorem \ref{thm:c-weakly} shows that a strong
equilibrium can be found by an improvement path of length at most
$3n+1$, possibly augmented by a single profitable deviation of all
players.  It would be interesting to extend this analysis of bounds on
the lengths of improvement paths to other cases when a Nash or a
strong equilibrium is known to exist.

In the future we plan to study the inefficiency of equilibria in
coordination games on directed graphs. Also, we plan to study 
coordination games on finite directed weighted graphs. While we
already defined here these games, we used weights solely as a means to
simplify the argument in the proof of Theorem~\ref{thm:NE-NPcomplete}.
It should be noted that Lemma \ref{lem:weakly} does not hold for finite directed
weighted graphs and, as a consequence, Theorems 
\ref{thm:c-weakly}, \ref{thm:SCCs}, and \ref{thm:linear-time-on-cycle} 
do not hold either.
A counterexample to Lemma \ref{lem:weakly} can be constructed 
by modifying the game in Figure \ref{fig:graph}
as follows. Nodes 4, 5, 6 are removed and replaced by assigning 
weight 2 to all the edges in the cycle. Nodes 7, 8, 9 are also removed and 
replaced by a +1 bonus to the colour of the node removed.
It is easy to see that the behaviour of this new game will 
mimic the game in Figure \ref{fig:graph}.
On the other hand, Theorem \ref{thm:DAG} and its proof is still valid for 
finite directed weighted graphs
as well is Theorem \ref{thm:NE-NPcomplete}, 
because checking whether a colour assignment 
is a Nash equilibrium can still 
be done in polynomial time for them.

As an example of coordination games on weighted directed graphs
consider the problem of a choice of the trade treaties between various
countries. Assume a directed weighted graph in which the nodes are the
countries and the weight on an edge $i \to j$ corresponds to the
percentage of the overall import of country $j$ from country $i$.
Suppose additionally that each country should choose a specific trade
treaty, that the options for the countries differ (for instance
because of its geographic location) and that each treaty offers the
same tax-free advantages. Then once the countries choose the treaties,
the payoff to each country is the aggregate percentage of its import
that is tax-free.

The case of weighted directed graphs can be seen as a minor
modification of the \bfe{social network games} with obligatory product
selection that we introduced and analyzed in \cite{AS13}.
These are games associated with a
threshold model of a social network introduced in \cite{AM11} which is
based on weighted graphs with thresholds. The difference consists of
using thresholds equal to 0. However, setting the
thresholds to 0 essentially changes the nature of the games and
crucially affects the validity of several arguments.

\subsection*{Acknowledgments}
We are grateful to Mona Rahn and Guido Sch\"{a}fer for useful
discussions and thank Piotr Sankowski and the referees for helpful
comments.  First author is also a Visiting Professor at the University
of Warsaw.  He was partially supported by the NCN grant nr
2014/13/B/ST6/01807. The last author is partially supported by EPSRC
grant EP/M027287/1.

%\nocite{*}
\bibliographystyle{eptcs}
\bibliography{e,clustering}

\begin{thebibliography}{10}
\providecommand{\bibitemdeclare}[2]{}
\providecommand{\surnamestart}{}
\providecommand{\surnameend}{}
\providecommand{\urlprefix}{Available at }
\providecommand{\url}[1]{\texttt{#1}}
\providecommand{\href}[2]{\texttt{#2}}
\providecommand{\urlalt}[2]{\href{#1}{#2}}
\providecommand{\doi}[1]{doi:\urlalt{http://dx.doi.org/#1}{#1}}
\providecommand{\bibinfo}[2]{#2}

\bibitemdeclare{inproceedings}{AM11}
\bibitem{AM11}
\bibinfo{author}{K.~R. \surnamestart Apt\surnameend} \&
  \bibinfo{author}{E.~\surnamestart Markakis\surnameend}
  (\bibinfo{year}{2011}): \emph{\bibinfo{title}{Diffusion in Social Networks
  with Competing Products}}.
\newblock In: {\sl \bibinfo{booktitle}{Proceedings of the 4th International
  Symposium on Algorithmic Game Theory (SAGT)}}, {\sl \bibinfo{series}{Lecture
  Notes in Computer Science}} \bibinfo{volume}{6982},
  \bibinfo{publisher}{Springer}, pp. \bibinfo{pages}{212--223},
  \doi{10.1007/978-3-642-24829-0\_20}.

\bibitemdeclare{inproceedings}{ARSS14}
\bibitem{ARSS14}
\bibinfo{author}{K.~R. \surnamestart Apt\surnameend},
  \bibinfo{author}{M.~\surnamestart Rahn\surnameend},
  \bibinfo{author}{G.~\surnamestart Sch{\"a}fer\surnameend} \&
  \bibinfo{author}{S.~\surnamestart Simon\surnameend} (\bibinfo{year}{2014}):
  \emph{\bibinfo{title}{Coordination Games on Graphs (extended abstract)}}.
\newblock In: {\sl \bibinfo{booktitle}{Proceedings of the 10th Conference on
  Web and Internet Economics (WINE)}}, {\sl \bibinfo{series}{Lecture Notes in
  Computer Science}} \bibinfo{volume}{8877}, \bibinfo{publisher}{Springer}, pp.
  \bibinfo{pages}{441--446}, \doi{10.1007/978-3-319-13129-0\_37}.

\bibitemdeclare{misc}{ARSS15}
\bibitem{ARSS15}
\bibinfo{author}{K.~R. \surnamestart Apt\surnameend},
  \bibinfo{author}{M.~\surnamestart Rahn\surnameend},
  \bibinfo{author}{G.~\surnamestart Sch{\"a}fer\surnameend} \&
  \bibinfo{author}{S.~\surnamestart Simon\surnameend} (\bibinfo{year}{2015}):
  \emph{\bibinfo{title}{Coordination Games on Graphs}}.
\newblock \bibinfo{note}{Available from \url{http://arxiv.org/abs/1501.07388}}.

\bibitemdeclare{inproceedings}{AS13}
\bibitem{AS13}
\bibinfo{author}{K.~R. \surnamestart Apt\surnameend} \&
  \bibinfo{author}{S.~\surnamestart Simon\surnameend} (\bibinfo{year}{2013}):
  \emph{\bibinfo{title}{Social Network Games with Obligatory Product
  Selection}}.
\newblock In: {\sl \bibinfo{booktitle}{Proceedings 8th International Symposium
  on Games, Automata, Logics and Formal Verification (GandALF)}},
  \bibinfo{volume}{119}, \bibinfo{publisher}{Electronic Proceedings in
  Theoretical Computer Science}, pp. \bibinfo{pages}{180--193},
  \doi{10.4204/EPTCS.119}.

\bibitemdeclare{incollection}{Aumann59}
\bibitem{Aumann59}
\bibinfo{author}{R.~J. \surnamestart Aumann\surnameend} (\bibinfo{year}{1959}):
  \emph{\bibinfo{title}{Acceptable Points in General Cooperative N-person
  Games}}.
\newblock In \bibinfo{editor}{R.~D. \surnamestart Luce\surnameend} \&
  \bibinfo{editor}{A.~W. \surnamestart Tucker\surnameend}, editors: {\sl
  \bibinfo{booktitle}{Contribution to the theory of game IV, Annals of
  Mathematical Study 40}}, \bibinfo{publisher}{University Press}, pp.
  \bibinfo{pages}{287--324}.

\bibitemdeclare{book}{DPV06}
\bibitem{DPV06}
\bibinfo{author}{S.~\surnamestart Dasgupta\surnameend}, \bibinfo{author}{C.H.
  \surnamestart Papadimitriou\surnameend} \& \bibinfo{author}{U.~\surnamestart
  Vazirani\surnameend} (\bibinfo{year}{2006}):
  \emph{\bibinfo{title}{Algorithms}}.
\newblock \bibinfo{publisher}{McGraw-Hill}.

\bibitemdeclare{article}{Harks13}
\bibitem{Harks13}
\bibinfo{author}{T.~\surnamestart Harks\surnameend},
  \bibinfo{author}{M.~\surnamestart Klimm\surnameend} \& \bibinfo{author}{R.H.
  \surnamestart M{\"o}hring\surnameend} (\bibinfo{year}{2013}):
  \emph{\bibinfo{title}{Strong Equilibria in Games with the Lexicographical
  Improvement Property}}.
\newblock {\sl \bibinfo{journal}{International Journal of Game Theory}}
  \bibinfo{volume}{42}(\bibinfo{number}{2}), pp. \bibinfo{pages}{461--482},
  \doi{10.1007/s00182-012-0322-1}.

\bibitemdeclare{article}{Holzman97}
\bibitem{Holzman97}
\bibinfo{author}{R.~\surnamestart Holzman\surnameend} \&
  \bibinfo{author}{N.~\surnamestart Law-Yone\surnameend}
  (\bibinfo{year}{1997}): \emph{\bibinfo{title}{Strong Equilibrium in
  Congestion Games}}.
\newblock {\sl \bibinfo{journal}{Games and Economic Behavior}}
  \bibinfo{volume}{21}(\bibinfo{number}{1-2}), pp. \bibinfo{pages}{85--101},
  \doi{10.1006/game.1997.0592}.

\bibitemdeclare{incollection}{JZ14}
\bibitem{JZ14}
\bibinfo{author}{M.~\surnamestart Jackson\surnameend} \&
  \bibinfo{author}{Y.~\surnamestart Zenou\surnameend} (\bibinfo{year}{2014}):
  \emph{\bibinfo{title}{Games on Networks}}.
\newblock In \bibinfo{editor}{H.~Peyton \surnamestart Young\surnameend} \&
  \bibinfo{editor}{Shmuel \surnamestart Zamir\surnameend}, editors: {\sl
  \bibinfo{booktitle}{Handbook of Game Theory 4}},
  \bibinfo{publisher}{Elsevier}, pp. \bibinfo{pages}{95--163},
  \doi{10.1016/B978-0-444-53766-9.00003-3}.

\bibitemdeclare{inproceedings}{KLS01}
\bibitem{KLS01}
\bibinfo{author}{M.~\surnamestart Kearns\surnameend},
  \bibinfo{author}{M.~\surnamestart Littman\surnameend} \&
  \bibinfo{author}{S.~\surnamestart Singh\surnameend} (\bibinfo{year}{2001}):
  \emph{\bibinfo{title}{Graphical Models for Game Theory}}.
\newblock In: {\sl \bibinfo{booktitle}{Proceedings of the 17th Conference in
  Uncertainty in Artificial Intelligence (UAI '01)}},
  \bibinfo{publisher}{Morgan Kaufmann}, pp. \bibinfo{pages}{253--260}.

\bibitemdeclare{article}{KBW97}
\bibitem{KBW97}
\bibinfo{author}{H.~\surnamestart Konishi\surnameend},
  \bibinfo{author}{M.~\surnamestart {Le Breton}\surnameend} \&
  \bibinfo{author}{S.~\surnamestart Weber\surnameend} (\bibinfo{year}{1997}):
  \emph{\bibinfo{title}{Equivalence of Strong and Coalition-proof {Nash}
  Equilibria in Games without Spillovers}}.
\newblock {\sl \bibinfo{journal}{Economic Theory}}
  \bibinfo{volume}{9}(\bibinfo{number}{1}), pp. \bibinfo{pages}{97--113},
  \doi{10.1007/BF01213445}.

\bibitemdeclare{article}{KBW97a}
\bibitem{KBW97a}
\bibinfo{author}{H.~\surnamestart Konishi\surnameend},
  \bibinfo{author}{M.~\surnamestart {Le Breton}\surnameend} \&
  \bibinfo{author}{S.~\surnamestart Weber\surnameend} (\bibinfo{year}{1997}):
  \emph{\bibinfo{title}{Pure Strategy {Nash} Equilibrium in a Group Formation
  Game with Positive Externalities}}.
\newblock {\sl \bibinfo{journal}{Games and Economic Behaviour}}
  \bibinfo{volume}{21}, pp. \bibinfo{pages}{161--182},
  \doi{10.1006/game.1997.0542}.

\bibitemdeclare{article}{Mil96}
\bibitem{Mil96}
\bibinfo{author}{I.~\surnamestart Milchtaich\surnameend}
  (\bibinfo{year}{1996}): \emph{\bibinfo{title}{Congestion Games with
  Player-Specific Payoff Functions}}.
\newblock {\sl \bibinfo{journal}{Games and Economic Behaviour}}
  \bibinfo{volume}{13}, pp. \bibinfo{pages}{111--124},
  \doi{10.1006/game.1996.0027}.

\bibitemdeclare{article}{SA15}
\bibitem{SA15}
\bibinfo{author}{S.~\surnamestart Simon\surnameend} \& \bibinfo{author}{K.~R.
  \surnamestart Apt\surnameend} (\bibinfo{year}{2015}):
  \emph{\bibinfo{title}{Social Network Games}}.
\newblock {\sl \bibinfo{journal}{Journal of Logic and Computation}}
  \bibinfo{volume}{25}(\bibinfo{number}{1}), pp. \bibinfo{pages}{207--242},
  \doi{10.1093/logcom/ext012}.

\bibitemdeclare{article}{Jan68}
\bibitem{Jan68}
\bibinfo{author}{E.B. \surnamestart Yanovskaya\surnameend}
  (\bibinfo{year}{1968}): \emph{\bibinfo{title}{Equilibrium Points in
  Polymatrix Games}}.
\newblock {\sl \bibinfo{journal}{Litovskii Matematicheskii Sbornik}}
  \bibinfo{volume}{8}, pp. \bibinfo{pages}{381--384}.

\bibitemdeclare{article}{You93}
\bibitem{You93}
\bibinfo{author}{H.~Peyton \surnamestart Young\surnameend}
  (\bibinfo{year}{1993}): \emph{\bibinfo{title}{The Evolution of Conventions}}.
\newblock {\sl \bibinfo{journal}{Econometrica}}
  \bibinfo{volume}{61}(\bibinfo{number}{1}), pp. \bibinfo{pages}{57--84},
  \doi{10.2307/2951778}.

\end{thebibliography}
\end{document}